\documentclass[pdflatex,sn-mathphys]{sn-jnl}

\jyear{2021}%
\usepackage{commath}
\usepackage{makecell}
\usepackage{tabu}

\theoremstyle{thmstyleone}%
\newtheorem{theorem}{Theorem}
\newtheorem{lemma}{{{\textit{Lemma}}}}

\theoremstyle{thmstyletwo}%
\newtheorem{example}{Example}%
\newtheorem{remark}{Remark}%

\theoremstyle{thmstylethree}%
\newtheorem{definition}{Definition}%
\newcounter{cases}
\newcounter{subcases}[cases]
\newenvironment{mycases}
  {%
    \setcounter{cases}{0}%
    \setcounter{subcases}{0}%
    \def\case
      {%
        \par\noindent
        \refstepcounter{cases}%
        \textbf{Case \thecases.}
      }%
    \def\subcase
      {%
        \par\noindent
        \refstepcounter{subcases}%
        \textit{Subcase (\thesubcases):}
      }%
  }
  {%
    \par
  }
\renewcommand*\thecases{\arabic{cases}}
\renewcommand*\thesubcases{\roman{subcases}}

\begin{document}

\title[A Direct Construction of Optimal 2D-ZCACS]{A Direct Construction of Optimal 2D-ZCACS with Flexible Array Size and Large Set Size}


\author[1]{\fnm{Gobinda} \sur{Ghosh}}\email{gobinda\_1921ma06@iitp.ac.in}

\author*[2]{\fnm{Sudhan} \sur{Majhi}}\email{smajhi@iisc.ac.in}

\author[1]{\fnm{Shubhabrata} \sur{Paul}}\email{shubhabrata@iitp.ac.in}

\affil[1]{\orgdiv{Mathematics}, \orgname{IIT Patna}, \orgaddress{\street{Bihta}, \city{Patna}, \postcode{801103}, \state{Bihar}, \country{India}}}

\affil*[2]{\orgdiv{Electrical Communication Engineering}, \orgname{IISc Bangalore}, \orgaddress{\street{CV Raman Rd}, \city{Bengaluru}, \postcode{560012}, \state{Karnataka}, \country{India}}}



\abstract{In this paper, we propose a direct construction of optimal two-dimensional Z-complementary array code sets (2D-ZCACS) using  multivariable functions (MVFs). In contrast to earlier works, the proposed construction allows for a flexible array size and a large set size. Additionally, the proposed design can be transformed into a one-dimensional Z-complementary code set (1D-ZCCS). Many of the 1D-ZCCS described in the literature appeared to be special cases of this proposed construction. At last, we compare our work with the current state of the art and then draw our conclusions.}




\maketitle

\section{Introduction}\label{sec1}
For an asynchronous two dimensional multi-carrier code-division multiple
access (2D-MC-CDMA) system, the ideal 2D correlation
properties of two dimensional complete
complementary codes (2D-CCCs)\cite{farkas2003two} can be properly utilized to obtain interference-free performance \cite{turcsany2004new}.
Similar to one dimensional complete complementary code (1D-CCC)\cite{chen2008complete,das2018novel,liu2014new}, one of the most significant drawbacks of 2D-CCC is that the set size is restricted \cite{xeng2004theoretical}. Motivated by the scarcity of 2D-CCC with flexible
set sizes, Zeng \textit{et al.} proposed 2D Z-complementary array code sets (2D-ZCACSs) in \cite{zeng2005construction,xeng2004theoretical}. For a $2D-(K,Z_{1} \times Z_{2})-\text{ZCACS}_{M}^{L_{1}\times L_{2}},K,Z_{1}\times Z_{2},L_{1}\times L_{2}$ and $M$ denote the set size,
two dimensional zero-correlation zone (2D-ZCZ) width, array size and the number of constituent arrays, respectively. In \cite{zeng2005construction,xeng2004theoretical}, authors obtained ternary 2D-ZCACSs by inserting some zeros into the existing binary 2D-ZCACSs. In 2021,
Pai \textit{et al}. presented a new construction method of 2D binary
Z-complementary array pairs (2D-ZCAP) \cite{pai2021two}.
Recently, Das \textit{et al}. in \cite{das2020two}
proposed a construction of 2D-ZCACS by using Z-paraunitary (ZPU) matrices. All these constructions of 2D-ZCACS depend heavily on initial sequences and matrices which increase hardware storage. For the first time in the literature, Roy \textit{et al}. in \cite{roy2021construction} proposed a direct construction of 2D-ZCACS based on MVF. The array size of the proposed 2D-ZCACS is of the form $L_{1}\times L_{2}$, where $L_{1}=2^{m}$, $L_{2}=2p_{1}^{m_{1}}p_{2}^{m_{2}}\ldots p_{k}^{m_{k}}$, $m\geq1, m_{i}\geq 2$ and the set size is of the form $2p_{1}^{2}p_{2}^{2}\ldots p_{k}^{2}$ where $p_{i}$ is a prime number. Therefore the array size and the set size is restricted to some even numbers.
\par Existing array and set size limitations through direct construction in the literature motivates us to search multivariable function (MVF) for more flexible array and set sizes. Our proposed construction provides 2D-ZCACS with parameter $2D-(R_{1}R_{2}M_{1}M_{2},N_{1} \times N_{2})-\text{ZCACS}_{M_{1}M_{2}}^{R_{1}N_{1}\times R_{2}N_{2}}$ where
$M_{1}=\prod_{i=1}^{a}p_{i}^{k_{i}}$, $M_{2}=\prod_{j=1}^{b}q_{j}^{t_{j}}$, $p_{i}$ is any prime or $1$, $q_{j}$ is prime, $a,b,k_{i},t_{j}\geq 1$,
$R_{1}$ and $R_{2}$ are positive integer, such that $R_{1}\geq 1$ and $R_{2}\geq 2$, $N_{1}=\prod_{i=1}^{a}p_{i}^{m_{i}}$, $N_{2}=\prod_{j=1}^{b}q_{j}^{n_{j}}$, $m_{i},n_{j}\geq 1$. The set size in our proposed 2D-ZCACS construction, $R_{1}R_{2}M_{1}M_{2}$, is more adaptable than the set size of 2D-ZCACS given in \cite{roy2021construction}. Unlike \cite{roy2021construction}, the proposed 2D-ZCACS can be reduced to 1D-ZCCS \cite{shen2022new,sarkar2020construction,sarkar2020direct,wu2020z,kumar2022direct,sarkar2018optimal,sarkar2021pseudo,ghosh2022direct} also. As a result, many existing optimal 1D-ZCCSs have become special cases of the proposed construction \cite{sarkar2018optimal,sarkar2021pseudo,ghosh2022direct}. The proposed construction also derived a new set of optimal 1D-ZCCS that had not previously been presented by direct method.
\par The rest of the paper is organized as follows. Section $2$ discusses construction related definitions and lemmas. Section $3$ contains the construction of 2D-ZCACS and the comparison with the existing state-of-the-art. Finally, in Section 4, the conclusions are drawn.
\section{Notations and definitions}
The following notations will be followed throughout this paper: $\omega_{n}=\exp\left(2\pi\sqrt{-1}/n\right)$, $\mathbb{A}_{n}=\{0,1,\ldots,n-1\}\subset \mathbb{Z}$, where $n$ is a positive integer and $\mathbb{Z}$ is the ring of integer.
\subsection{Two Dimensional Array}	
\begin{definition}[\cite{das2020two}]
Let $\mathbf{A}=\left(a_{g, i}\right)$ and $\mathbf{B}=\left(b_{g, i}\right)$ be complex-valued arrays of size $l_{1} \times l_{2}$ where $0 \leq g<l_{1},0 \leq i<l_{2}$.
The two dimensional aperiodic cross correlation function (2D-ACCF) of arrays $\mathbf{A}$ and $\mathbf{B}$ at shift $\left(\tau_{1}, \tau_{2}\right)$ is defined as
   \begin{equation*}
   \begin{split}
       \boldsymbol{C}\left(\mathbf{A},\mathbf{B}\right)\left(\tau_{1}, \tau_{2}\right)=
\begin{cases}
\sum_{g=0}^{l_{1}-1-\tau_{1}}\sum_{i=0}^{l_{2}-1-\tau_{2}} a_{g, i}b^{*}_{g+\tau_{1}, i+\tau_{2}},\text{if}~~\makecell{0 \leq \tau_{1}<l_{1},\\0 \leq \tau_{2}<l_{2};}\\
\sum_{g=0}^{l_{1}-1-\tau_{1}}\sum_{i=0}^{l_{2}-1+\tau_{2}} a_{g, i-\tau_{2}}b^{*}_{g+\tau_{1}, i},\text{if}~~\makecell{0 \leq \tau_{1}<l_{1},\\ -l_{2} < \tau_{2}<0;}\\
\sum_{g=0}^{l_{1}-1+\tau_{1}}\sum_{i=0}^{l_{2}-1-\tau_{2}} a_{g-\tau_{1}, i}b^{*}_{g, i+\tau_{2}},\text{if}~~\makecell{-l_{1} < \tau_{1}<0,\\0 \leq \tau_{2}<l_{2};}\\
\sum_{g=0}^{l_{1}-1+\tau_{1}}\sum_{i=0}^{l_{2}-1+\tau_{2}} a_{g-\tau_{1},i-\tau_{2}}b^{*}_{g,i},\text{if}~~\makecell{-l_{1} <\tau_{1}<0,\\-l_{2} < \tau_{2}<0.}
\end{cases}
   \end{split}
\end{equation*}
\end{definition}
Here, $(.)^{*}$ denotes the complex conjugate. If $\mathbf{A}=\mathbf{B},$ then $ \boldsymbol{C}\left(\mathbf{A}, \mathbf{B}\right)\left(\tau_{1},\tau_{2}\right)$ is called the two dimensional aperiodic auto correlation function (2D-AACF) of $\mathbf{A}$ and referred to as $\boldsymbol{C}\left(\mathbf{A}\right)\left(\tau_{1},\tau_{2}\right)$.
\par When $l_{1}= 1$, the complex-valued arrays $\mathbf{A}$ and $\mathbf{B}$ are reduced to one dimensional complex-valued sequences $\mathbf{A}=(a_{j})_{j=0}^{l_{2}-1}$ and $\mathbf{B}=(b_{j})_{j=0}^{l_{2}-1}$ with the corresponding one dimensional aperiodic cross correlation function (1D-ACCF) given by
\begin{equation}\label{equ:cross}
\boldsymbol{C}(\mathbf{A},\mathbf{B})({\tau_{2}})=\begin{cases}
\sum_{i=0}^{l_{2}-1-\tau_{2}}a_{i}b^{*}_{i+\tau_{2}}, & 0 \leq \tau_{2} < l_{2}, \\
\sum_{i=0}^{l_{2}+\tau_{2} -1}a_{i-\tau_{2}}b^{*}_{i}, & -l_{2}< \tau_{2} < 0,  \\
0, & \text{otherwise}.
\end{cases}
\end{equation}
\begin{definition}\cite{pai2022designing},\cite{das2020two}
    For a set of $s$ sets of arrays $\boldsymbol{A}=\left\{\mathbf{A}^{k} \mid k=\right.$ $0,1, \ldots, s-1\}$, each set $\mathbf{A}^{k}=\left\{\mathbf{A}_{0}^{k}, \mathbf{A}_{1}^{k}, \ldots, \mathbf{A}_{s-1}^{k}\right\}$ is composed of $s$ arrays of size is $l_{1} \times l_{2}$. The set $\boldsymbol{A}$ is said to be 2D-CCC with parameters $(s,s,l_{1},l_{2})$ if the following holds
\begin{equation}
  \begin{split}
\boldsymbol{C}\left(\mathbf{A}^{k},\mathbf{A}^{k^{\prime}}\right)\left(\tau_{1},\tau_{2}\right)&=\sum_{i=0}^{s-1} \boldsymbol{C}\left(\mathbf{A}_{i}^{k}, \mathbf{A}_{i}^{k^{\prime}}\right)\left(\tau_{1}, \tau_{2}\right)\\
&= \begin{cases}sl_{1}l_{2}, \quad\left(\tau_{1}, \tau_{2}\right)=(0,0), k=k^{\prime} ;\\
0, \quad\left(\tau_{1}, \tau_{2}\right)\neq(0,0), k=k^{\prime} ;\\
0,     ~~~~k\neq k^{\prime}.
\end{cases}
  \end{split} 
\end{equation}
\end{definition}
\begin{definition}\cite{roy2021construction},\cite{das2020two}
Let $z_1, z_2, l_{1}, l_{2}$ are positive integers and $z_{1}\leq l_{1}, z_{2}\leq l_{2}$. Consider the  sets of $\hat{s}$ set of arrays $\boldsymbol{A}=\left\{\mathbf{A}^{k} \mid k=\right.$ $0,1, \ldots, \hat{s}-1\}$, where each set $\mathbf{A}^{k}=\left\{\mathbf{A}_{0}^{k}, \ldots, \mathbf{A}_{s-1}^{k}\right\}$ is composed of $s$ arrays of size $l_{1} \times l_{2}$. The set $\boldsymbol{A}$ is said to be $2D-(\hat{s},z_{1}\times z_{2})-\text{ZCACS}_{s}^{l_{1}\times l_{2}}$ if the following holds
\begin{equation}
    \begin{split} \boldsymbol{C}\left(\mathbf{A}^{k},\mathbf{A}^{k^{\prime}}\right)\left(\tau_{1},\tau_{2}\right)&=\sum_{i=0}^{s-1} \boldsymbol{C}\left(\mathbf{A}_{i}^{k}, \mathbf{A}_{i}^{k^{\prime}}\right)\left(\tau_{1}, \tau_{2}\right)\\
        &= \begin{cases}sl_{1}l_{2}, \quad\left(\tau_{1}, \tau_{2}\right)=(0,0), k=k^{\prime} ;\\
0, \quad\left(\tau_{1}, \tau_{2}\right)\neq(0,0),\abs{\tau_{1}}<z_{1},\abs{\tau_{2}}<z_{2}, k=k^{\prime} ;\\
0,     ~~~~\abs{\tau_{1}}<z_{1},\abs{\tau_{2}}<z_{2},k\neq k^{\prime}.
\end{cases}
    \end{split}
\end{equation}
\end{definition}
When $z_{1}=l_{1}, z_{2}=l_{2},\hat{s}=s$ the 2D-ZCACS becomes 2D-CCC\cite{ghosh2022direct1,pai2022designing} with parameter $(s,l_{1},l_{2})$.
It should be noted that for $l_1=1$, each array $\mathbf{A}_{i}^{k}$ becomes  $l_2$-length sequence. Therefore, 2D-ZCACS can be reduced to a conventional 1D-$\left(\hat{s}, z_{2}\right)- \textit{ZCCS}_{s}^{l_{2}}$\cite{wu2018optimal}, \cite{yu2022new},\cite{shen2022new11}, where, $\hat{s},s,z_{2},l_{2}$ represents no. of set, set size, ZCZ width and sequence length respectively. 
\begin{lemma}
\cite{das2020two}
For a $2D-(\hat{s},z_{1}\times z_{2})-\text{ZCACS}_{s}^{l_{1}\times l_{2}}$, the following inequality holds
\begin{equation}
    \hat{s}z_{1}z_{2}\leq s\left(l_{1}+z_{1}-1\right)\left( l_{2}+z_{2}-1\right).
\end{equation}
We called 2D-ZCACS is optimal if the following equality holds
\begin{equation}
\label{318}
    \hat{s}=s\Big\lfloor\frac{l_{1}}{z_{1}}\Big\rfloor\Big\lfloor\frac{l_{2}}{z_{2}}\Big\rfloor,
\end{equation}
where $\lfloor.\rfloor$ denotes the floor function.
\end{lemma}
\subsection{Multivariable Function}	
Let $a$, $b$, $m_i$, and $n_j$ be positive integers for $1\leq i\leq a$ and $1\leq j\leq b$. Let $p_{i}$ be any prime or $1$, and $q_{j}$ be a prime number. 
A multivariable function (MVF) can be defined as 
\begin{equation*}
    f: \mathbb{A}_{p_1}^{m_1}\times \mathbb{A}_{p_2}^{m_2}\times \dots \times \mathbb{A}_{p_a}^{m_a}\times\mathbb{A}_{q_1}^{n_1}\times \mathbb{A}_{q_2}^{n_2} \times \dots \times \mathbb{A}_{q_b}^{n_b}\rightarrow\mathbb{Z}.
\end{equation*}
 Let $c,d\geq0$ be integers such that $0\leq c <r$ and  $0\leq d<s$ where $r=p_{1}^{m_{1}}p_{2}^{m_{2}}\ldots p_{a}^{m_{a}}$ and $s=q_{1}^{n_{1}}q_{2}^{n_{2}}\ldots q_{b}^{n_{b}}$. Then $c$ and $d$ can be written as
 \begin{equation}
 \label{c,d}
    \begin{split}
&c=c_{1}+c_{2}p_{1}^{m_{1}}+\dots+c_{a}p_{1}^{m_{1}}p_{2}^{m_{2}}\ldots p_{a-1}^{m_{a-1}},\\
&d=d_{1}+d_{2}q_{1}^{n_{1}}+\dots+d_{b}q_{1}^{n_{1}}q_{2}^{n_{2}}\ldots q_{b-1}^{n_{b-1}},
    \end{split}
\end{equation}
where, $0\leq c_{i}< p_{i}^{m_{i}}$ and $0\leq d_{j}< q_{j}^{n_{j}}$. 
Let $\mathbf{C}_{i}=(c_{i,1},c_{i,2},\ldots,c_{i,m_{i}})\in \mathbb{A}_{p_{i}}^{m_{i}}$,
be the vector representation of $c_{i}$ with base $p_{i}$, i.e., $c_{i}=\sum_{k=1}^{m_{i}}c_{i,k}p_{i}^{k-1}$ and $\mathbf{D}_{j}=(d_{j,1},d_{j,2},\ldots,d_{j,n_{j}})\in \mathbb{A}_{q_{j}}^{n_{j}}$ be the vector representation of  $d_{j}$  with base $q_{j}$,  i.e., $d_{j}=\sum_{l=1}^{n_{j}}d_{j,l}q_{j}^{l-1}$ where $0\leq c_{i,k}<p_{i}$, and $0\leq d_{j,l}< q_{j}$. We define vectors associated with $c$ and $d$ as
\begin{equation*}
\begin{split}
&\phi(c)=\left(\mathbf{C}_{1},\mathbf{C}_{2},\ldots,\mathbf{C}_{a}\right)\in \mathbb{A}_{p_1}^{m_1}\times \mathbb{A}_{p_2}^{m_2}\times \dots \times \mathbb{A}_{p_a}^{m_a},\\
&\phi(d)=\left(\mathbf{D}_{1},\mathbf{D}_{2},\ldots,\mathbf{D}_{b}\right)\in \mathbb{A}_{q_1}^{n_1}\times \mathbb{A}_{q_2}^{n_2} \times \dots \times \mathbb{A}_{q_b}^{n_b},
\end{split}    
\end{equation*}
 respectively. We also define an array associated with $f$ as
\begin{equation}
    \psi_{\lambda}({f})=\left(\begin{array}{cccc}
\omega_{\lambda}^{f_{0,0}} & \omega_{\lambda}^{f_{0,1}} & \cdots & \omega_{\lambda}^{f_{0,r-1}} \\
\omega_{\lambda}^{f_{1,0}} & \omega_{\lambda}^{f_{1,1}} & \cdots & \omega_{\lambda}^{f_{1,r-1}} \\
\vdots & \vdots & \ddots & \vdots \\
\omega_{\lambda}^{f_{s-1,0}} & \omega_{\lambda}^{f_{s-1,1}} & \cdots & \omega_{\lambda}^{f_{s-1,r-1}}
\end{array}\right),
\end{equation}
where $f_{c,d}=f\left(\phi(c),\phi(d)\right)$ and $\lambda$ is a positive integer.
\begin{lemma}[\cite{vaidyanathan2014ramanujan}]
\label{DauJi}
Let $t$ and $t'$ be two non-negative integers, where $t\neq t'$, and $p$ is a prime number. Then
\begin{equation}
    \displaystyle\sum_{j=0}^{p-1}\omega_{p}^{(t-t')j}=0.
\end{equation}
\end{lemma}
Let us consider the set $\mathcal{C}$ as
\begin{equation}
    \mathcal{C}=\left(\mathbb{A}_{p_1}^{m_1}\times \mathbb{A}_{p_2}^{m_2}\times \dots \times \mathbb{A}_{p_a}^{m_a}\right)\times\left(\mathbb{A}_{q_1}^{n_1}\times \mathbb{A}_{q_2}^{n_2} \times \dots \times \mathbb{A}_{q_b}^{n_b}\right).
\end{equation}
Let  $0\leq \gamma< p_{1}^{m_{1}}p_{2}^{m_{2}}\ldots p_{a}^{m_{a}}$ and $0\leq \mu< q_{1}^{n_{1}}q_{2}^{n_{2}}\ldots q_{b}^{n_{b}}$ be positive integers such that
\begin{equation}
\begin{split}
&\gamma=\gamma_{1}+\sum_{i=2}^{a}\gamma_{i}\left(\prod_{i_{1}=1}^{i-1}p_{i_{1}}^{m_{i_{1}}}\right),\\
&\mu=\mu_{1}+\sum_{j=2}^{b}\mu_{j}\left(\prod_{j_{1}=1}^{j-1}q_{j_{1}}^{n_{j_{1}}}\!\right),
\end{split}
\end{equation}
where $0\leq \gamma_{i}< p_{i}^{m_{i}}$ and $0\leq \mu_{j}< q_{j}^{n_{j}}$.
Let $\boldsymbol{\gamma}_{i}=(\gamma_{i,1},\gamma_{i,2},\ldots,\gamma_{i,m_{i}})\in \mathbb{A}_{p_{i}}^{m_{i}}$ be the vector representation of $\gamma_{i}$ with base $p_{i}$, i.e.,  $\gamma_{i}=\sum_{k=1}^{m_{i}}\gamma_{i,k}p_{i}^{k-1}$, where $0\leq \gamma_{i,k}<p_{i}$. Similarly
$\boldsymbol{\mu}_{j}=(\mu_{j,1},\mu_{j,2},\ldots,\mu_{j,n_{j}})\in \mathbb{A}_{q_{j}}^{n_{j}}$ be the vector representation of $\mu_{j}$ with base $q_{j}$ i.e., $\mu_{j}=\sum_{l=1}^{n_{j}}\mu_{j,l}q_{j}^{l-1}$ where $0\leq \mu_{j,l}<q_{j}$. Let 
\begin{equation}
\phi(\gamma)=\left(\boldsymbol{\gamma}_{1},\boldsymbol{\gamma}_{2},\ldots,\boldsymbol{\gamma}_{a}\right)\in \mathbb{A}_{p_1}^{m_1}\!\!\times \!\mathbb{A}_{p_2}^{m_2}\!\times \dots \times \mathbb{A}_{p_a}^{m_a},   
\end{equation}
be the vector associated with $\gamma$ and
\begin{equation}
\phi(\mu)=\left(\boldsymbol{\mu}_{1},\boldsymbol{\mu}_{2},\ldots,\boldsymbol{\mu}_{b}\right)\in \mathbb{A}_{q_1}^{n_1}\!\!\times \!\mathbb{A}_{q_2}^{n_2}\!\times \dots \times \mathbb{A}_{q_b}^{n_b},   
\end{equation}
 be the vector associated with $\mu$.
Let
$\pi_{i}$ and $\sigma_{j}$ be any permutations of the set  $\{1,2,\ldots,m_{i}\}$ and  $\{1,2,\ldots,n_{j}\}$, respectively.
Let us also define the MVF $f:\mathcal{C}\rightarrow\mathbb{Z},$ as \begin{equation}
\label{a11111}
    \begin{split}
           &f(\phi(\gamma),\phi(\mu))\\
&=f\left(\boldsymbol{\gamma}_{{1}},\boldsymbol{\gamma}_{{2}}, \ldots, \boldsymbol{\gamma}_{{a}},\boldsymbol{\mu}_{{1}},\boldsymbol{\mu}_{{2}}, \ldots, \boldsymbol{\mu}_{b}\right)\\
           &=\sum_{i=1}^{a}\!\frac{\lambda}{p_{i}}\!\!\sum_{e=1}^{m_{i}-1}\!\!\gamma_{i, \pi_{i}(e)} \gamma_{i, \pi_{i}(e+1)}+\sum_{i=1}^{a}\!\sum_{e=1}^{m_{i}}\!d_{i,e} \gamma_{i, e}
+\sum_{j=1}^{b}\!\frac{\lambda}{q_{j}}\!\!\sum_{o=1}^{n_{j}-1} \mu_{j, \sigma_{j}(o)} \mu_{j, \sigma_{j}(o+1)}\\
&+\sum_{j=1}^{b}\!\sum_{o=1}^{n_{j}} c_{j,o} \mu_{j,o},
    \end{split}
\end{equation}
where $d_{i,e},c_{j,o}\in \{0,1,\ldots,\lambda-1\}$ and $\lambda=l.c.m.(p_{1},$ $\ldots,p_{a},q_{1},\ldots,q_{b})$. Let us define the set $\Theta$ and $T$ as
   \begin{equation*}
  \label{kaka}
      \begin{split}
       &\Theta=\{\theta:\theta=(r_{{1}}, r_{{2}}, \ldots, r_{{a}},s_{{1}}, s_{{2}}, \ldots, s_{{b}})\},\\
       &T=\{t:t=(x_{{1}}, x_{{2}}, \ldots, x_{{a}},y_{{1}}, y_{{2}}, \ldots, y_{{b}})\},
      \end{split}
  \end{equation*}
  where $0\leq r_{i},x_{i}< p_{i}^{k_{i}}$ and $0\leq s_{j},y_{j}< q_{j}^{r_{j}}$ and $k_{i},r_{j}$ are positive integers.
  Now, we define a function 	$a^{\theta}_{t}\!\!:\mathcal{C}\rightarrow\!\! \mathbb{Z},$ as
  \begin{equation}{\label{5}}
 \begin{split}
 &a^{\theta}_{t}\left(\phi(\gamma),\phi(\mu)\right)\\
 &=a^{\theta}_{t}\left(\boldsymbol{\gamma}_{{1}},\boldsymbol{\gamma}_{{2}}, \ldots, \boldsymbol{\gamma}_{{a}},\boldsymbol{\mu}_{{1}},\boldsymbol{\mu}_{{2}}, \ldots, \boldsymbol{\mu}_{b}\right)\\
     &=\!\!f\left(\phi(\gamma),\phi(\mu)\right)\!\!+\!\!\sum_{i=1}^{a} \frac{\lambda}{p_{i}} \gamma_{i, \pi_{i}(1)}{r_{i}}+\!\sum_{j=1}^{b} \frac{\lambda}{q_{j}} \mu_{j, \sigma_{j}(1)}{s_{j}}
    +\sum_{i=1}^{a} \frac{\lambda}{p_{i}} \gamma_{i, \pi_{i}(m_{i})}{x_{i}}\\
    &+\sum_{j=1}^{b} \frac{\lambda}{q_{j}} \mu_{j, \sigma_{j}(n_{j})}{y_{j}}+d_{\theta},
 \end{split}
\end{equation}
where $0\leq d_{\theta}<\lambda$, $\gamma_{i, \pi_{i}(1)},\gamma_{i, \pi_{i}(m_{i})}$ denote $\pi_{i}(1)-$th and  $\pi_{i}(m_{i})-$th element of
$\boldsymbol{\gamma}_{{i}}$ respectively.  Similarly, $\mu_{j, \sigma_{j}(1)},\mu_{j, \sigma_{j}(n_{j})}$ denote $\sigma_{j}(1)-$th and $\sigma_{j}(n_{j})-th$ element of $\boldsymbol{\mu}_{{j}}$ respectively.  For simplicity, we denote $a^{\theta}_{t}\left(\phi({\gamma}),\phi({\mu})\right)$ by $(a^{\theta}_{t})_{\gamma,\mu}$ and $f\left(\phi({\gamma}),\phi({\mu})\right)$ by $f_{\gamma,\mu}$.
\begin{lemma}[\cite{ghosh2022direct1}]
\label{KB}
 We define the ordered set of arrays $\mathbf{A}^{t}=\{\psi_{\lambda}\left(a^{\theta}_{t}\right):\theta\in \Theta\}$. Then the set  $\{\mathbf{A}^{t}:t\in T\}$ forms a 2D-CCC with parameter
$
(\alpha,\alpha,m,n)
$, where, $\alpha=\prod_{i=1}^{a}p^{k_{i}}_{i}\prod_{j=1}^{b}q^{r_{j}}_{j}$, $m=\prod_{i=1}^{a}p_{i}^{m_{i}}$, $n=\prod_{j=1}^{b}q_{j}^{n_{j}}$ and $k_{i},m_{i},n_{j},r_{j}$ are non-negative integers.
\end{lemma}
\section{Proposed construction of 2D-ZCACS}
Let $a',b'$ be positive integers for $1\leq i'\leq a'$ and $1\leq j'\leq b'$,~
$p_{i'}'$ be any prime or $1$, and
$q_{j'}'$ be prime number. Let $\gamma',\mu'$ are positive integers such that $0\leq \gamma'<\left(\prod_{i=1}^{a}p_{i}^{m_{i}}\right)\left(\prod_{i'=1}^{a'}p'_{i'}\right)$ and $0\leq \mu'< \left(\prod_{j=1}^{b}q_{j}^{n_{j}}\right)\left(\prod_{j'=1}^{b'}q'_{j'}\right)$. Then  $\gamma',\mu'$ can be written as
\begin{equation}
     \begin{split}
         &\gamma'\!=\!\gamma_{1}\!+\!\!\displaystyle\sum_{i=2}^{a}\gamma_{i}\left(\prod_{i_{1}=1}^{i-1}p_{i_{1}}^{m_{i_{1}}}\right)\!\!+\!\!\left(\gamma'_{1}+\sum_{i'=2}^{a'}\gamma'_{i'}\left(\prod_{i_{1}=1}^{i'-1}p'_{i_{1}}\right)\right)
        m,\\
         &\mu'\!=\!\mu_{1}\!+\!\!\displaystyle\sum_{j=2}^{b}\mu_{j}\!\!\left(\prod_{j_{1}=1}^{j-1}q_{j_{1}}^{n_{j_{1}}}\right)\!\!+\!\!\left(\mu'_{1}+\sum_{j'=2}^{b'}\mu'_{j'}\left(\prod_{j_{1}=1}^{j'-1}q'_{j_{1}}\right)\right)
         n,
     \end{split}
 \end{equation}
where $m=\prod_{i=1}^{a}p_{i}^{m_{i}}$, $n=\prod_{j=1}^{b}q_{j}^{n_{j}}$, $0\leq\gamma_{i}<p_{i}^{m_{i}}$, $0\leq\mu_{j}<q_{j}^{n_{j}}$, $0\leq\gamma_{i'}'<p_{i'}'$ and $0\leq\mu_{j'}'< q_{j'}'$. We denote the vectors associated with $\gamma'$ and $\mu'$ are
\begin{equation}
\begin{split}
&\phi(\gamma')=\left(\boldsymbol{\gamma}_{{1}}, \ldots, \boldsymbol{\gamma}_{{a}},\gamma_{1}',\ldots,\gamma_{a}'\right)\in\mathbb{A}_{p_{1}}^{m_{1}}\times\hdots\times\mathbb{A}_{p_{a}}^{m_{a}}\times\mathbb{A}_{p'_{1}}\times\hdots\times\mathbb{A}_{p'_{a'}},\\
&\phi(\mu')=\left(\boldsymbol{\mu}_{{1}}, \ldots, \boldsymbol{\mu}_{{b}},\mu_{1}',\ldots,\mu_{b}'\right)\in \mathbb{A}_{q_{1}}^{n_{1}}\times\hdots\times\mathbb{A}_{q_{b}}^{n_{b}}\times\mathbb{A}_{q'_{1}}\times\hdots\times\mathbb{A}_{q'_{b'}},
\end{split}    
\end{equation}
respectively, where $\boldsymbol{\gamma}_{{i}}\in \mathbb{A}_{p_i}^{m_i}$, $\boldsymbol{\mu}_{{j}}\in\mathbb{A}_{q_j}^{n_j}$ 
are the vectors associated with $\gamma_{i}$ and $\mu_{j}$ respectively i.e.,
$\boldsymbol{\gamma}_{i}=(\gamma_{i,1},\gamma_{i,2},\ldots,\gamma_{i,m_{i}})\in \mathbb{A}_{p_{i}}^{m_{i}}$, $\boldsymbol{\mu}_{j}=(\mu_{j,1},\mu_{j,2},\ldots,\mu_{j,n_{j}})\in \mathbb{A}_{q_{j}}^{n_{j}}$ , $\gamma_{i}=\sum_{k=1}^{m_{i}}\gamma_{i,k}p_{i}^{k-1}$, $\mu_{j}=\sum_{l=1}^{n_{j}}\mu_{i,l}q_{j}^{l-1}$, $0\leq \gamma_{i,k}<p_{i}$ and $0\leq \mu_{j,l}<q_{j}$.
Let us consider the set $\mathcal{D}$ as
\begin{equation}
  \mathcal{D}= \mathbb{A}_{p_{1}}^{m_{1}}\times\hdots\times\mathbb{A}_{p_{a}}^{m_{a}}\times\mathbb{A}_{p'_{1}}\times\hdots\times\mathbb{A}_{p'_{a'}}\times\mathbb{A}_{q_{1}}^{n_{1}}\times\hdots\times\mathbb{A}_{q_{b}}^{n_{b}}\times\mathbb{A}_{q'_{1}}\times\hdots\times\mathbb{A}_{q'_{b'}}.
\end{equation}
Let $f$ be the function as defined (\ref{a11111}). We define the MVF
 $M^{\mathbf{c},\mathbf{d}}:\mathcal{D}\rightarrow \mathbb{Z}$  as
 \begin{equation}
 \begin{split}
 \label{rama}
     &M^{\mathbf{c},\mathbf{d}}\left(\phi(\gamma'),\phi(\mu')\right)\\
     &=M^{\mathbf{c},\mathbf{d}}\left(\boldsymbol{\gamma}_{{1}}, \ldots, \boldsymbol{\gamma}_{{a}},\gamma_{1}',\ldots,\gamma_{a'}',\boldsymbol{\mu}_{{1}}, \ldots, \boldsymbol{\mu}_{{b}},\mu_{1}',\ldots,\mu_{b'}'\right)\\
     &=\frac{\delta}{\lambda}f\left(\boldsymbol{\gamma}_{{1}}, \ldots, \boldsymbol{\gamma}_{{a}},\boldsymbol{\mu}_{{1}}, \ldots, \boldsymbol{\mu}_{b}\right)\!+\!\!\sum_{i'=1}^{a'}\!c_{i'}\frac{\delta}{p'_{i'}}\gamma_{i'}'+\!\!\sum_{j'=1}^{b'}\!d_{j'}\frac{\delta}{q'_{j'}}\mu_{j'}',
     \end{split}
 \end{equation}
 where $0\leq c_{i'}< p'_{i'}$, $0\leq d_{j'}<q'_{j'}$, $\mathbf{c}=(c_{1},c_{2},\ldots
,c_{a'})$ and $\mathbf{d}=(d_{1},d_{2},\ldots,d_{b'})$. For simplicity, now on-wards we denote $M^{\mathbf{c},\mathbf{d}}(\boldsymbol{\gamma}_{{1}}, \ldots, \boldsymbol{\gamma}_{{a}},\gamma_{1}',\ldots,\gamma_{a'}',\boldsymbol{\mu}_{{1}}, \ldots, \boldsymbol{\mu}_{{b}},\mu_{1}',\ldots,\mu_{b'}')$ by $M^{\mathbf{c},\mathbf{d}}$. 
Consider the set $\Theta$ and $T$ as
   \begin{equation*}
      \begin{split}
       &\Theta=\{\theta:\theta=(r_{{1}}, r_{{2}}, \ldots, r_{{a}},s_{{1}}, s_{{2}}, \ldots, s_{{b}})\},\\
       &T=\{t:t=(x_{{1}}, x_{{2}}, \ldots, x_{{a}},y_{{1}}, y_{{2}}, \ldots, y_{{b}})\},
      \end{split}
  \end{equation*}
  where $0\leq r_{i},x_{i}< p_{i}^{k_{i}}$ and $0\leq s_{j},y_{j}< q_{j}^{r_{j}}$ and $k_{i},r_{j}$ are positive integers.
 Let us define MVF, $b_{t}^{\theta,\mathbf{c},\mathbf{d}}:\mathcal{D}\rightarrow \mathbb{Z}$, as
\begin{equation}
\label{hare}
    \begin{split}
b_{t}^{\theta,\mathbf{c},\mathbf{d}}=&M^{\mathbf{c},\mathbf{d}}+\sum_{i=1}^{a} \frac{\delta}{p_{i}} \gamma_{i, \pi_{i}(1)}{r_{i}}+\sum_{j=1}^{b} \frac{\delta}{q_{j}} \mu_{j, \sigma_{j}(1)}{s_{j}}+\sum_{i=1}^{a} \frac{\delta}{p_{i}} \gamma_{i, \pi_{i}(m_{i})}{x_{i}}\\
&+\sum_{j=1}^{b} \frac{\delta}{q_{j}} \mu_{j, \sigma_{j}(n_{j})}{y_{j}}+\frac{\delta}{\lambda}d_{\theta},
    \end{split}
\end{equation}
where $0\leq d_{\theta}<\lambda$.
By (\ref{5}), (\ref{rama}) and (\ref{hare}) we have
\begin{equation}
    b_{t}^{\theta,\mathbf{c},\mathbf{d}}=\frac{\delta}{\lambda}a^{\theta}_{t}+\sum_{i'=1}^{a'}c_{i'}\frac{\delta}{p'_{i'}}\gamma'_{i'}+\sum_{j'=1}^{b'}d_{j'}\frac{\delta}{q'_{j'}}\mu'_{j'}.
\end{equation}
We define the ordered set of arrays as 
\begin{equation}
   \Omega_{t}^{\mathbf{c},\mathbf{d}}=\{\psi_{\delta}(b_{t}^{\theta,\mathbf{c},\mathbf{d}}):\theta\in \Theta\}.
\end{equation}
where $\delta=l.c.m(\lambda,p'_{1},p'_{2},\ldots,p'_{a'},q'_{1},q'_{2},\ldots,q'_{b'}).$
\begin{theorem}
\label{VrindavanBihari}
 Let $m=\prod_{i=1}^{a}p^{m_{i}}_{i}, n=\prod_{j=1}^{b}q^{n_{j}}_{j},\mathbf{c}=(c_{1},\ldots,c_{a'}),\mathbf{d}=(d_{1},\ldots,d_{b'})$. Then the set  $S=\{\Omega_{t}^{\mathbf{c},\mathbf{d}}:t\in T,0\leq c_{i'}<p'_{i'},0\leq d_{j'}<q'_{j'}\}$ forms a  $2D-(\alpha_{1},z_{1}\times z_{2})-\text{ZCACS}_{\alpha}^{l_{1}\times l_{2}}$, where,  $\alpha_{1}=\left(\prod_{i'=1}^{a'}p'_{i'}\right)\left(\prod_{j'=1}^{b'}q'_{j'}\right)\alpha$, $l_{1}=m\left(\prod_{i'=1}^{a'}p'_{i'}\right)$, $l_{2}=n\left(\prod_{j'=1}^{b'}q'_{j'}\right)$, $z_{1}=m$ ,$z_{2}=n$, $\alpha=(\prod_{i=1}^{a}p^{k_{i}}_{i})(\prod_{j=1}^{b}q^{r_{j}}_{j})$, $k_{i},r_{j},m_{i},n_{j}\geq 1$.
\end{theorem}
\begin{proof}
 Let $\hat{\gamma},\hat{\mu}$ are positive integers such that $0\leq\hat{\gamma}<l_{1}$ and  $0\leq\hat{\mu}<l_{2}$. Then $\hat{\gamma},\hat{\mu}$ can be written as
 \begin{equation*}
     \begin{split}
         &\hat{\gamma}=\gamma_{1}\!+\!\!\displaystyle\sum_{i=2}^{a}\gamma_{i}\left(\prod_{i_{1}=1}^{i-1}p_{i_{1}}^{m_{i_{1}}}\right)\!\!+\!\!\left(\gamma'_{1}+\sum_{i'=2}^{a'}\gamma'_{i'}\left(\prod_{i_{1}=1}^{i'-1}p'_{i_{1}}\right)\right)
        m,\\
         &\hat{\mu}=\mu_{1}\!+\!\!\displaystyle\sum_{j=2}^{b}\mu_{j}\left(\prod_{j_{1}=1}^{j-1}q_{j_{1}}^{n_{j_{1}}}\right)\!\!+\!\!\left(\mu'_{1}+\sum_{j'=2}^{b'}\mu'_{j'}\left(\prod_{j_{1}=1}^{j'-1}q'_{j_{1}}\right)\!\!\right)
         n,
     \end{split}
 \end{equation*}
 where $0\leq \gamma_{i}< p_{i}^{m_{i}}$, $0\leq \mu_{j}< q_{j}^{n_{j}}$, $0\leq \gamma'_{i'}<p'_{i'}$ and $0\leq \mu'_{j'}<q'_{j'}$. The proof will be split into following cases
 \begin{mycases}
     \case $(\tau_{1}=0,\tau_{2}=0)$\\

 The ACCF between $\Omega_{t}^{\mathbf{c},\mathbf{d}}$ and $\Omega_{t'}^{\mathbf{c}',\mathbf{d}'}$ at $\tau_{1}=0$ and $\tau_{2}=0$ can be expressed as 
 \begin{equation}
 \label{2.a}
     \begin{split}
    &C(\Omega_{t}^{\mathbf{c},\mathbf{d}},\Omega_{t'}^{\mathbf{c}',\mathbf{d}'})(0,0)\\
        &=\sum_{\theta\in \Theta}C(\psi_{\delta}((b_{t}^{\theta,\mathbf{c},\mathbf{d}})),\psi_{\delta}((b_{t'}^{\theta,\mathbf{c}',\mathbf{d}'})))(0,0)\\
        &=\sum_{\theta\in \Theta}\sum_{\hat{\gamma}=0}^{l_{1}-1}\sum_{\hat{\mu}=0}^{l_{2}-1}\omega_{\delta}^{(b_{t}^{\theta,\mathbf{c},\mathbf{d}})_{\hat{\gamma},\hat{\mu}}-(b_{t'}^{\theta,\mathbf{c}',\mathbf{d}'})_{\hat{\gamma},\hat{\mu}}}\\
        &=\sum_{\theta\in \Theta}\sum_{\gamma=0}^{m-1}\sum_{\mu=0}^{n-1}\sum_{\gamma'_{1}=0}^{p'_{1}-1}\ldots\sum_{\gamma'_{a'}=0}^{p'_{a'}-1}\sum_{\mu_{1}=0}^{q'_{1}-1}\ldots\sum_{\mu'_{b'}=0}^{q'_{b'}-1}\omega_{\delta}^{D},
        \end{split}
        \end{equation}
        where $D=\frac{\delta}{\lambda}\left((a_{t}^{\theta})_{\gamma,\mu}-(a_{t'}^{\theta})_{\gamma,\mu}\right)+\sum_{i'=1}^{a'}\frac{\delta}{p'_{i'}}(c_{i'}-c_{i'}')\gamma_{i'}+\sum_{j'=1}^{b'}\frac{\delta}{q'_{j'}}(d_{j'}-d_{j'}')\mu_{j'}$. After splitting (\ref{2.a}), we get
        \begin{equation}
        \label{029}
        \begin{split}
    &C(\Omega_{t}^{\mathbf{c},\mathbf{d}},\Omega_{t'}^{\mathbf{c}',\mathbf{d}'})(0,0)\\&=\left(\sum_{\theta\in \Theta}\sum_{\gamma=0}^{m-1}\sum_{\mu=0}^{n-1}\omega_{\delta}^{\frac{\delta}{\lambda}\left((a_{t}^{\theta})_{\gamma,\mu}-(a_{t'}^{\theta})_{\gamma,\mu}\right)}\right)\mathcal{E}\mathcal{F}\\
       &=\left(\sum_{\theta\in \Theta}\sum_{\gamma=0}^{m-1}\sum_{\mu=0}^{n-1}\omega_{\lambda}^{\left((a_{t}^{\theta})_{\gamma,\mu}-(a_{t'}^{\theta})_{\gamma,\mu}\right)}\right)\mathcal{E}\mathcal{F}\\
       &=C(\mathbf{A}^{t},\mathbf{A}^{t'})(0,0)\mathcal{E}\mathcal{F},
     \end{split}
 \end{equation}
  where 
  \begin{equation}
  \label{jajya}
    \begin{split}
&\mathcal{E}=\prod_{i'=1}^{a'}\left(\sum_{\gamma'_{i'}=0}^{p'_{i'}-1}\omega_{p'_{i'}}^{(c_{i'}-c_{i'}')\gamma'_{i'}}\right),\\ 
&\mathcal{F}=\prod_{j'=1}^{b'}\left(\sum_{\mu'_{j'}=0}^{q'_{j'}-1}\omega_{q'_{j'}}^{(d_{j'}-d'_{j'})\mu'_{j'}}\right).
    \end{split}  
  \end{equation}
 \subcase $(t\neq t'$)\\
By \textit{lemma} \ref{DauJi} we know, the set $\{\mathbf{A}^{t}:t\in T\}$ forms a 2D-CCC. Hence By \textit{lemma} \ref{DauJi}, we have
\begin{equation}
\label{sunnk}
    C(\mathbf{A}^{t},\mathbf{A}^{t'})(0,0)=0.
\end{equation}
Hence by (\ref{029}) and (\ref{sunnk}) we have
\begin{equation}
\label{king}
    C(\Omega_{t}^{\mathbf{c},\mathbf{d}},\Omega_{t'}^{\mathbf{c}',\mathbf{d}'})(0,0)=0.
\end{equation}
\subcase $(t= t'$)\\
 By \textit{lemma} \ref{DauJi}, we know 
\begin{equation}
\label{Hare}
    C(\mathbf{A}^{t},\mathbf{A}^{t'})(0,0)=\left(\prod_{i=1}^{a}p_{i}^{m_{i}+k_{i}}\right)\left(\prod_{j=1}^{b}q_{j}^{n_{j}+r_{j}}\right).
\end{equation}
Let $M=\left(\prod_{i=1}^{a}p_{i}^{m_{i}+k_{i}}\right)\left(\prod_{j=1}^{b}q_{j}^{n_{j}+r_{j}}\right)$ hence by  \textit{Lemma} \ref{DauJi}, (\ref{029}), (\ref{jajya}), (\ref{Hare}),  we have the following \begin{equation}
 \label{r.1}
 \begin{split}
     C(\Omega_{t}^{\mathbf{c},\mathbf{d}},\Omega_{t}^{\mathbf{c}',\mathbf{d}'})(0,0)
     =\begin{cases}
   M\left(\prod_{i'=1}^{a'}p'_{i'}\right)\left(\prod_{j'=1}^{b'}q'_{j'}\right)
    & {\mathbf{c}}={\mathbf{c}}',{\mathbf{d}}={\mathbf{d}}'\\
     0         ,&  {\mathbf{c}}\neq{\mathbf{c}}',{\mathbf{d}}={\mathbf{d}}'\\\
      0         ,&  {\mathbf{c}}={\mathbf{c}}',{\mathbf{d}}\neq{\mathbf{d}}'\\
      0         ,&  {\mathbf{c}}\neq{\mathbf{c}}',{\mathbf{d}}\neq{\mathbf{d}}'.
     \end{cases}
 \end{split}
 \end{equation}
\case$(0<\tau_1<\prod_{i=1}^{a}p_{i}^{m_{i}},0<\tau_2<\prod_{j=1}^{b}q_{j}^{n_{j}})$\\
Let $\sigma,\rho$ are positive integers such that $0\leq \sigma< m'$ and  $0\leq \rho< n'$ where $m'=\prod_{i'=1}^{a'}p'_{i'},n'=\prod_{j'=1}^{b'}q'_{j'}$. Then $\sigma$ and $\rho$ can be written as
\begin{equation}
    \begin{split}
       &\sigma=\sigma_{1}+\sigma_{2}p'_{1}+\ldots+\sigma_{a'}\left(\prod_{i'=1}^{a'-1}p'_{i'}\right),\\ 
&\rho=\rho_{1}+\rho_{2}q'_{1}+\ldots+\rho_{b'}\left(\prod_{j'=1}^{b'-1}q'_{j'}\right),
    \end{split}
\end{equation}
respectively where $0\leq \sigma_{i'}< p'_{i'}$ and $0\leq \rho_{j'}< q'_{j'}$ . We define vectors associated with $\sigma$ and $\rho$ to be 
\begin{equation}
\begin{split}
    &\phi(\sigma)=(\sigma_{1},\ldots,\sigma_{a'})\in \mathbb{A}_{p'_{1}}\times\hdots\times\mathbb{A}_{p'_{a'}},\\
    &\phi(\rho)=(\rho_{1},\ldots,\rho_{b'})\in \mathbb{A}_{q'_{1}}\times\hdots\times\mathbb{A}_{q'_{b'}},
\end{split}  
\end{equation}
respectively.
The ACCF between $\Omega_{t}^{\mathbf{c},\mathbf{d}}$ and  $\Omega_{t'}^{\mathbf{c}',\mathbf{d}'}$ for $0<\tau_1<\prod_{i=1}^{a}p_{i}^{m_{i}}$ and  $0<\tau_2<\prod_{j=1}^{b}q_{j}^{n_{j}}$, can be derived as \begin{equation}
\label{31}
    \begin{split}
        &C(\Omega_{t}^{c,d},\Omega_{t'}^{c',d'})(\tau_1,\tau_2)
        =\!C(\mathbf{A}^{t},\mathbf{A}^{t'})(\tau_1,\tau_2)DE\!+\!
        C(\mathbf{A}^{t},\mathbf{A}^{t'})(\tau_1\!-\!\prod_{i=1}^{a}p_{i}^{m_{i}},\tau_2)D'E+\\
        &C(\mathbf{A}^{t},\mathbf{A}^{t'})(\tau_1,\tau_2-\prod_{j=1}^{b}q_{j}^{n_{j}})DE'
        +C(\mathbf{A}^{t},\mathbf{A}^{t'})(\tau_1-\prod_{i=1}^{a}p_{i}^{m_{i}},\tau_2-\prod_{j=1}^{b}q_{j}^{n_{j}})D'E',
    \end{split}
\end{equation}
where \begin{eqnarray}
\label{sri}
D=\sum_{\sigma=0}^{m'-1}\left(\prod_{i'=1}^{a'}\omega_{p'_{i'}}^{(c_{i'}-c_{i'}')(\sigma_{i'})}\right),\\
      E=\sum_{\rho=0}^{n'-1}
    \left(\prod_{j'=1}^{b'}\omega_{q'_{j'}}^{(d_{j'}-d_{j'}')(\rho_{j'})}\right),~\\      D'=\displaystyle\sum_{\sigma=0}^{m'-2}\left(\prod_{i'=1}^{a'}\omega_{p'_{i'}}^{\left(c_{i'}(\sigma_{i'})-c'_{i'}\left(\sigma+1\right)_{i'}\right)}\right),\\
E'=\displaystyle\sum_{\rho=0}^{n'-2}\left(\prod_{j'=1}^{b'}\omega_{q'_{j'}}^{\left(d_{j'}(\rho_{j'})-d'_{j'}\left(\rho+1\right)_{j'}\right)}\right),
\end{eqnarray}
and
$\left(\sigma+1\right)_{i'},\left(\rho+1\right)_{j'}$ denotes the $i'$-th and $j'$-th components of $\phi\left(\sigma+1\right)$ and $\phi\left(\rho+1\right)$ respectively. 
By \textit{Lemma} \ref{DauJi}, for $0<\tau_1<\prod_{i=1}^{a}p_{i}^{m_{i}}$ and $0<\tau_2<\prod_{j=1}^{b}q_{j}^{n_{j}}$, we have 
\begin{eqnarray}
\label{3.1}
    C(\mathbf{A}^{t},\mathbf{A}^{t'})(\tau_1,\tau_2)=0,\\
    \label{3.2}
    C(\mathbf{A}^{t},\mathbf{A}^{t'})(\tau_1\!-\!\prod_{i=1}^{a}p_{i}^{m_{i}},\tau_2)=0,\\
    \label{3.3}
    C(\mathbf{A}^{t},\mathbf{A}^{t'})(\tau_1,\tau_2-\prod_{j=1}^{b}q_{j}^{n_{j}})=0,\\
    \label{3.4}
      C(\mathbf{A}^{t},\mathbf{A}^{t'})(\tau_1-\prod_{i=1}^{a}p_{i}^{m_{i}},\tau_2-\prod_{j=1}^{b}q_{j}^{n_{j}})=0.
\end{eqnarray}
By (\ref{31}), (\ref{3.1}), (\ref{3.2}), (\ref{3.3}), (\ref{3.4}) we have
\begin{equation}
\label{r.2}
    C(\Omega_{t}^\mathbf{c,d},\Omega_{t^{'}}^{\mathbf{c}^{'},\mathbf{d}^{'}})(\tau_1,\tau_2)=0.
\end{equation}
\case $(0<\tau_1<\prod_{i=1}^{a}p_{i}^{m_{i}},-\prod_{j=1}^{b}q_{j}^{n_{j}}<\tau_2<0)$\\
 The ACCF between $\Omega_{t}^{\mathbf{c},\mathbf{d}}$ and  $\Omega_{t'}^{\mathbf{c}',\mathbf{d}'}$ for $0<\tau_1<\prod_{i=1}^{a}p_{i}^{m_{i}}$ and  $-\prod_{j=1}^{b}q_{j}^{n_{j}}<\tau_2<0$, can be derived as

 \begin{equation}
 \label{muk}
    \begin{split}
        &C(\Omega_{t}^{\mathbf{c},\mathbf{d}},\Omega_{t'}^{\mathbf{c'},\mathbf{d'}})(\tau_1,\tau_2)\\
        &=\!C(\mathbf{A}^{t},\mathbf{A}^{t'})(\tau_1,\tau_2)DE\!+\!
        C(\mathbf{A}^{t},\mathbf{A}^{t'})(\tau_{1}-\prod_{i=1}^{a}p_{i}^{m_{i}},\tau_2)D'E\\
        &+
        C(\mathbf{A}^{t},\mathbf{A}^{t'})(\tau_1,\prod_{j=1}^{b}q_{j}^{n_{j}}+\tau_{2})DE''+C(\mathbf{A}^{t},\mathbf{A}^{t'})(\tau_1-\prod_{i=1}^{a}p_{i}^{m_{i}},\prod_{j=1}^{b}q_{j}^{n_{j}}+\tau_2)D'E'',
    \end{split}
\end{equation}
where
\begin{equation}
    \begin{split}
E''=\displaystyle\sum_{\rho=0}^{n'-2}\left(\prod_{j'=1}^{b'}\omega_{q'_{j'}}^{\left(d_{j'}\left(\rho+1\right)_{j'}-d_{j'}'(\rho_{j'})\right)}\right).
    \end{split}
\end{equation}
By \textit{Lemma} \ref{DauJi}, for $0<\tau_1<\prod_{i=1}^{a}p_{i}^{m_{i}}$ and $-\prod_{j=1}^{b}q_{j}^{n_{j}}<\tau_2<0$, we have 
\begin{eqnarray}
\label{1012}
C(\mathbf{A}^{t},\mathbf{A}^{t'})(\tau_1,\prod_{j=1}^{b}q_{j}^{n_{j}}+\tau_{2})=0,\\
\label{10121}
C(\mathbf{A}^{t},\mathbf{A}^{t'})(\tau_1-\prod_{i=1}^{a}p_{i}^{m_{i}},\prod_{j=1}^{b}q_{j}^{n_{j}}+\tau_2)=0.
\end{eqnarray}
By (\ref{muk}) , (\ref{1012}) and (\ref{10121}) we have
\begin{equation}
\label{r.3}
    C(\Omega_{t}^{\mathbf{c},\mathbf{d}},\Omega_{t'}^{\mathbf{c'},\mathbf{d'}})(\tau_1,\tau_2)=0.
\end{equation}
\case ($0<\tau_1<\prod_{i=1}^{a}p_{i}^{m_{i}},\tau_2=0$) \\

 The ACCF between $\Omega_{t}^{\mathbf{c},\mathbf{d}}$ and  $\Omega_{t'}^{\mathbf{c}',\mathbf{d}'}$ for $0<\tau_1<\prod_{i=1}^{a}p_{i}^{m_{i}}$ and $\tau_2=0$ , can be derived as
\begin{equation}
\label{308123}
    \begin{split}
        C(\Omega_{t}^{\mathbf{c},\mathbf{d}},\Omega_{t'}^{\mathbf{c'},\mathbf{d'}})(\tau_1,0)=\!C(\mathbf{A}^{t},\mathbf{A}^{t'})(\tau_1,0)DE\!+
        C(\mathbf{A}^{t},\mathbf{A}^{t'})(\tau_1-\prod_{i=1}^{a}p_{i}^{m_{i}},0)D'E.
    \end{split}
\end{equation}
By \textit{Lemma} \ref{DauJi}, for $0<\tau_1<\prod_{i=1}^{a}p_{i}^{m_{i}}$, we have 
 \begin{equation}
 \begin{split}
 \label{308000}
    &C(\mathbf{A}^{t},\mathbf{A}^{t'})(\tau_1,0)=0.\\
    & C(\mathbf{A}^{t},\mathbf{A}^{t'})(\tau_1-\prod_{i=1}^{a}p_{i}^{m_{i}},0)=0,
    \end{split}
\end{equation}
by (\ref{308123}) and (\ref{308000}) we have
\begin{equation}
    C(\Omega_{t}^{\mathbf{c},\mathbf{d}},\Omega_{t'}^{\mathbf{c'},\mathbf{d'}})(\tau_1,0)=0.
\end{equation} 
\case \\($\tau_1=0,0<\tau_2<\prod_{j=1}^{b}q_{j}^{n_{j}}$)\\

 The ACCF between $\Omega_{t}^{\mathbf{c},\mathbf{d}}$ and  $\Omega_{t'}^{\mathbf{c}',\mathbf{d}'}$ for $\tau_1=0$ and $0<\tau_2<\prod_{j=1}^{b}q_{j}^{n_{j}}$, can be derived as
\begin{equation}
\label{3081}
    \begin{split}
        C(\Omega_{t}^{\mathbf{c},\mathbf{d}},\Omega_{t'}^{\mathbf{c'},\mathbf{d'}})(0,\tau_2)=\!C(\mathbf{A}^{t},\mathbf{A}^{t'})(0,\tau_2)DE\!+
        C(\mathbf{A}^{t},\mathbf{A}^{t'})(0,\tau_2-\prod_{j=1}^{b}q_{j}^{n_{j}})DE'.
    \end{split}
\end{equation}
By \textit{Lemma} \ref{DauJi}, for $0<\tau_2<\prod_{j=1}^{b}q_{j}^{n_{j}}$, we have 
 \begin{equation}
 \begin{split}
 \label{3080}
    &C(\mathbf{A}^{t},\mathbf{A}^{t'})(0,\tau_2)=0,\\
    & C(\mathbf{A}^{t},\mathbf{A}^{t'})(0,\tau_2-\prod_{j=1}^{b}q_{j}^{n_{j}})=0.
    \end{split}
\end{equation}
By (\ref{3081}) and (\ref{3080}) we have
\begin{equation}
    C(\Omega_{t}^{\mathbf{c},\mathbf{d}},\Omega_{t'}^{\mathbf{c'},\mathbf{d'}})(0,\tau_2)=0.
\end{equation} 
\case($\tau_1=0,-\prod_{j=1}^{b}q_{j}^{n_{j}}<\tau_2<0$)\\
Similarly the ACCF between $\Omega_{t}^{\mathbf{c},\mathbf{d}}$ and  $\Omega_{t'}^{\mathbf{c}',\mathbf{d}'}$ for $\tau_1=0$ and $-\prod_{j=1}^{b}q_{j}^{n_{j}}<\tau_2<0$ is
\end{mycases}
\begin{equation}
\label{0912}
    \begin{split}
        C(\Omega_{t}^{\mathbf{c},\mathbf{d}},\Omega_{t'}^{\mathbf{c'},\mathbf{d'}})(0,\tau_2)
        =\!C(\mathbf{A}^{t},\mathbf{A}^{t'})(0,\tau_2)DE\!+
        C(\mathbf{A}^{t},\mathbf{A}^{t'})(0,\tau_2+\prod_{j=1}^{b}q_{j}^{n_{j}})DE''.
    \end{split}
\end{equation}
By \textit{Lemma} \ref{DauJi}, for $-\prod_{j=1}^{b}q_{j}^{n_{j}}<\tau_2<0$, we have 
 \begin{equation}
 \begin{split}
 \label{30818}
    & C(\mathbf{A}^{t},\mathbf{A}^{t'})(0,\tau_2+\prod_{j=1}^{b}q_{j}^{n_{j}})=0.
    \end{split}
\end{equation}
Hence by (\ref{3080}), (\ref{0912}) and (\ref{30818})  we have
\begin{equation}
\label{r.5}
    C(\Omega_{t}^{\mathbf{c},\mathbf{d}},\Omega_{t'}^{\mathbf{c'},\mathbf{d'}})(0,\tau_2)=0.
\end{equation} 
Combining all the cases  we have
\begin{equation}
\label{r.9}
\begin{split}
C(\Omega_{t}^{\mathbf{c},\mathbf{d}},\Omega_{t'}^{\mathbf{c'},\mathbf{d'}})(\tau_{1},\tau_2)
  =
\begin{cases}
M\left(\prod_{i'=1}^{a'}p'_{i'}\right)\left(\prod_{j'=1}^{b'}q'_{j'}\right),
    &\makecell{({\mathbf{c}},{\mathbf{d}},t)=({\mathbf{c'}},{\mathbf{d'}},t')\\(\tau_{1},\tau_{2})=(0,0),}\\
    \\
    \\
    0,
&\makecell{({\mathbf{c}},{\mathbf{d}},t)\neq({\mathbf{c'}},{\mathbf{d'}},t')\\(\tau_{1},\tau_{2})=(0,0),}\\
    \\
0,
    &\makecell{0\leq \tau_{1}<\prod_{i=1}^{a}p^{m_{i}}_{i},\\
    (\tau_{1},\tau_{2})\neq(0,0).}
\end{cases}
\end{split}
\end{equation}
Similarly it can be shown
\begin{equation}
\label{r.10}
    C(\Omega_{t}^{\mathbf{c},\mathbf{d}},\Omega_{t'}^{\mathbf{c'},\mathbf{d'}})(\tau_{1},\tau_2)=0,~-\prod_{i=1}^{a}p^{m_{i}}_{i}<\tau_{1}<0.
\end{equation}
Hence from (\ref{r.9}), (\ref{r.10})
we derive our conclusion.
\end{proof}
\begin{example}
Suppose that $a=1$, $b=1$, $a'=1$, $b'=1$, $p_1=2$, $m_1=2$, $k_{1}=1$, $q_1=3$, $n_1=2$, $r_{1}=1$, $p_{1}'=3$, $q_{1}'=2$. Let $\delta=6,~\lambda=6$, $\boldsymbol{\gamma}_{1}=(\gamma_{11},\gamma_{12})\in\mathbb{A}^{2}_{2}=\{0,1\}^{2}$ be the vector associated with ${\gamma}_{1}$ where $0\leq \gamma_{1}\leq 3$, i.e., $\gamma_{1}=\gamma_{11}+2\gamma_{12}$  and $\boldsymbol{\mu}_{1}=(\mu_{11},\mu_{12})\in\mathbb{A}_{3}^{2}=\{0,1,2\}^{2}$ be the vector associated with $\mu_{1}$ where $0\leq\mu_{1}\leq 8$, i.e., $\mu_{1}=\mu_{11}+3\mu_{12}$ and  $0\leq \gamma'_{1}\leq 2$, $0\leq \mu'_{1}\leq 1$.
We define the MVF $f:\mathbb{A}_{2}^{2}\times \mathbb{A}_{3}^{2}\rightarrow \mathbb{Z}$ as
\begin{equation*}
    \begin{split}
        f\left(\boldsymbol{\gamma}_{1},\boldsymbol{\mu}_{1}\right)\!\!=&3\gamma_{1,2}\gamma_{1,1}\!\!+\!\gamma_{1,1}\!+\!2\gamma_{1,2}\!+\!2\mu_{1,2}\mu_{1,1}\!+\!2\mu_{1,1}\!+\!\mu_{1,2}.
    \end{split}
\end{equation*}
Consider the MVF, $M^{\mathbf{c},\mathbf{d}}:\mathbb{A}_{2}^{2}\times\mathbb{A}_{3}\times \mathbb{A}_{3}^{2}\times\mathbb{A}_{2} \rightarrow \mathbb{Z}$ as
\begin{equation}
\begin{split}
    &M^{\mathbf{c},\mathbf{d}}\left(\boldsymbol{\gamma}_{1},\gamma'_{1},\boldsymbol{\mu}_{1},\mu'_{1}\right)\\
    &=f(\boldsymbol{\gamma}_{1},\boldsymbol{\mu}_{1})+2c_{1}\gamma'_{1}+3d_{1}\mu'_{1}\\
    &=3\gamma_{1,2}\gamma_{1,1}+\gamma_{1,1}+2\gamma_{1,2}+2\mu_{1,2}\mu_{1,1}+2\mu_{1,1}+\mu_{1,2}+2c_{1}\gamma'_{1}+3d_{1}\mu'_{1},
\end{split}
\end{equation}
 where  $0\leq c_{1} <p'_{1}=2$, $0 \leq d_{1} < q'_{1}=3$, $\mathbf{c}=c_{1}\in\{0,1\}, \text{and}~ \mathbf{d}=d_{1}\in\{0,1,2\}$. We have 
 \begin{equation}
      \begin{split}
       &\Theta=\{\theta:\theta=(r_{{1}}, s_{{1}}):0\leq r_1\leq 1,0\leq s_1\leq 2\},\\
       &T=\{t:t=(x_{{1}}, y_{{1}}):0\leq x_1\leq 1,0\leq y_1\leq 2\}.
      \end{split}
  \end{equation}
Let $d_{\theta}=0$, now from (\ref{hare}) we have
 \begin{equation}
     b_{t}^{\theta,\mathbf{c},\mathbf{d}}=M^{\mathbf{c},\mathbf{d}}+3\gamma_{1,2}r_{1}+2\mu_{1,2}s_{1}+3\gamma_{1,1}x_{1}+2\mu_{1,2}y_{1},
 \end{equation}
 and
\begin{equation}
\begin{split}
 \Omega_t^{\mathbf{c},\mathbf{d}}\!\!=\!\!\Big\{&
  \psi_{6}(b_{t}^{\theta,\mathbf{c},\mathbf{d}}):\theta=(r_{1},s_{1})\in\{0,1\}\times\{0,1,2\}\Big\}.
 \end{split}
\end{equation}
Therefore, the set
\begin{equation}
    S=\{ \Omega_t^{\mathbf{c},\mathbf{d}}:t\in T,0\leq c_{1}\leq 1,0\leq d_{1}\leq 2 \},
\end{equation}
 forms an optimal $2D-(36,4\times 9)-\text{ZCACS}_{6}^{12\times 18}$ over $\mathbb{Z}_{6}$. 
\end{example}
\begin{table}[h]
\begin{center}
\begin{minipage}{\textwidth}
\caption{Comparison with Previous Works}\label{tab2}
\begin{tabular*}{\textwidth}{@{\extracolsep{\fill}}lcccccc@{\extracolsep{\fill}}}
\toprule%
 \\%
Source & No. of set &  Array Size &  Condition & Based on    \\
\midrule
\cite{zeng2005construction} & $K=K'r$   & $L'_{1}\!\!\times\!\!(L'_{2}+r+1)$ & $r\geq 0$ & $\makecell{2D-ZCACS~of \\set~size~K'~
and \\array~size~L'_{1}\!\!\times\!\!L'_{2}}$ \\
\\
\cite{pai2021two}  & 1 & $2^{m}\times 2^{n}L$  & $m,n\geq 0$ & ZCP \!of length \!$L$ \\
\cite{das2020two}  & $K$ & $K\times K$& \makecell{$K$ divides set size} & BH matrices \\
\cite{roy2021construction} & $2\prod_{i=1}^{k_{i}}p^{2}_{i}$  & $2^{m} \times \prod_{i=1}^{k_{i}}p_{i}^{m_{i}}$& $k_{i},m_{i}\geq 1$, $p_{i}$'s are prime & MVF \\
\\
\textit{Thm 2}  & $rs\alpha$  & $rm\times sn$ & $\makecell { \alpha=(\prod_{i=1}^{a}p^{k_{i}}_{i})(\prod_{j=1}^{b}q^{r_{j}}_{j}),\\m\!\!=\!\!\prod_{i=1}^{a}\!\!p^{m_{i}}_{i},n\!\!=\!\!\prod_{j=1}^{b}\!\!q^{n_{j}}_{j},\\ r,s,\alpha\geq 1, p_{i},q_{j} are primes}$& MVF\\
\botrule
\end{tabular*}
\end{minipage}
\end{center}
\end{table}
\begin{remark}
In \textit{Theorem \ref{VrindavanBihari}},  if we take $a=1$, $p_{1}=1$, $a'=1$,  $p_{1}'=1$, $b=1$, $q_{1}=2$, $b'=l$, $r_{1}\geq 2$, we have optimal 1D-ZCCS with parameter  $(\prod_{i=1}^{l}q'_{i}2^{r_{1}},2^{n_{1}})-\text{ZCCS}_{2^{r_{1}}}^{\prod_{i=1}^{l}q'_{i}2^{n_{1}}}$, which is exactly the same result as in \cite{ghosh2022direct}. Also if we take $l=1$, then we have optimal $1$D-ZCCS of the form $(q'_{1}2^{r_{1}},2^{n_{1}})-\text{ZCCS}_{2^{r_{1}}}^{q'_{1}2^{n_{1}}}$, which is exactly the same result in \cite{sarkar2021pseudo}. Therefore the optimal $1$D-ZCCS given by \cite{ghosh2022direct,sarkar2021pseudo} appears as a special case of the proposed construction
\end{remark}
\begin{remark}
In \textit{Theorem \ref{VrindavanBihari}},  if  $a=1$, $p_{1}=1$, $a'=1$,  $p_{1}'=1$, $b=1$, $q_{1}=2$, $b'=l$, $r_{1}=1$, we have 1d-ZCCS with parameter  $(2\prod_{i=1}^{l}q'_{i},2^{n_{1}})-\text{ZCCS}_{2}^{\prod_{i=1}^{l}q'_{i}2^{n_{1}}}$, which is just a collection of $2\prod_{i=1}^{l}q'_{i}$ ZCPs with sequence length $\prod_{i=1}^{l}q'_{i}2^{n_{1}}$ and ZCZ width $2^{n_{1}}$. Hence our work produces collections of ZCPs\cite{kumar2022direct} as well.
\end{remark}

\begin{remark}
In \textit{Theorem \ref{VrindavanBihari}}, if we take
$a=1$, $p_{1}=1$, $a'=1$, $p_{1}'=1$, $b=1$, $q_{1}=2$, $b'=r$, $q'_{1}=q'_{2}=\ldots=q'_{r}=2$, $n_{1}=m-r$~\text{and}~$r_{1}=s+1$ then we have 1D-ZCCS with parameter $(2^{s+r+1},2^{m-r})-\text{ZCCS}_{2^{s+1}}^{2^{m}}$, which is exactly the same result in \cite{sarkar2018optimal}. Hence, the ZCCS in \cite{sarkar2018optimal} appears as a special case of our proposed construction.
\end{remark}
\begin{remark}
The 2D-ZCACS given by the proposed construction satisfies the equality given in (\ref{318}). Therefore the 2D-ZCACS obtained by the proposed construction is optimal. 
\end{remark}
\begin{remark}
If we take $a=1$, $a'=1$, $p_{1}=1$ and $p'_{1}=1$, in \textit{Theorem \ref{VrindavanBihari}}, we have optimal 1D-ZCCS with parameter $\left(\left(\prod_{j'=1}^{b'}q'_{j'}\right)\prod_{j=1}^{b}q^{r_{j}}_{j},n\right)-\text{ZCCS}_{\prod_{j=1}^{b}q^{r_{j}}_{j}}^{n\left(\prod_{j'=1}^{b'}q'_{j'}\right)}$ where, $n=\prod_{j=1}^{b}q_{j}^{n_{j}}$. Hence, we have optimal 1D-ZCCS of length $nm$ where, $n,m>1$ and $m=\prod_{j'=1}^{b'}q'_{j'}$. Therefore our construction produces optimal 1D-ZCCS
with a new length which is not present in the literature by direct method.
\end{remark}
\begin{remark}
    The set size of our proposed 2D-ZCACS is $\left(\prod_{i'=1}^{a'}p'_{i'}\right)\left(\prod_{j'=1}^{b'}q'_{j'}\right)\prod_{i=1}^{a}p^{k_{i}}_{i}\prod_{j=1}^{b}q^{r_{j}}_{j}$ where, $k_{i},t_{j}\geq 1$. If we take $a=1,p_{1}=1,a'=1,p'_{1}=1,r_{1}=r_{2}=\ldots=r_{b}=2,b'=1,~\text{and}~q'_{1}=2$ then we have set size $2\prod_{j=1}^{b}q^{2}_{j}$ which is the set size of the 2D-ZCACS in \cite{roy2021construction}. Therefore, we have flexible number of set sizes compared to \cite{roy2021construction}. 
\end{remark}
\subsection{Comparison with Previous Works}
Table I compares the proposed work with indirect constructions from \cite{zeng2005construction,das2020two,pai2021two} and direct construction from \cite{roy2021construction}. The constructions in \cite{zeng2005construction,das2020two,pai2021two} heavily rely on initial sequences, increasing hardware storage. The construction in \cite{roy2021construction} is direct, but set size and array sizes are limited to some even numbers. Our construction doesn't require initial matrices or sequences and produces flexible parameters.
\section{Conclusion}
In this paper, 2D-ZCACSs are designed by using MVF. The proposed design does not depend on initial sequences or matrices, so it is direct. Our proposed design produces flexible array size and set size compared to existing works. Also, our proposed construction can be reduced to 1D-ZCCS. As a result, many 1D-ZCCSs become special cases of our work. Finally, we compare our work to the existing state-of-the-art and show that it's more versatile.
\bibliography{sn-bibliography}


\end{document}